\newif\ifsubmit\newif\ifphone %

\submittrue

\ifsubmit %
\documentclass[11pt]{article} %
\else %
\ifphone %
\documentclass[11pt]{article} %
\else
\documentclass{article} %
\fi
\fi

\usepackage{article} %
\usepackage{math} %
\usepackage{algo} %
\usepackage{wrapfig} %
\usepackage[section]{placeins} %
\usepackage{graphicx} %
\usepackage{layout}
\usepackage[longnamesfirst,numbers]{natbib} %
\usepackage{setspace}           %
\usepackage[center,big]{titlesec} %

\usepackage[font={small},labelfont=bf,textfont=it]{caption}

\ifsubmit %
\usepackage[letterpaper]{geometry} 
\else                        %
\ifphone                     %
\usepackage[paperwidth=6in,paperheight=12.99in,margin=1cm,includehead,includefoot]{geometry}
\fi
\fi

\everymath{\displaystyle}

\begin{document}

\title{Approximating optimal~transport with~linear~programs}

\author{Kent Quanrud\thanks{University of Illinois, Urbana-Champaign;
    \href{mailto:quanrud2@illinois.edu}{\nolinkurl{quanrud2@illinois.edu};}
    \url{http://illinois.edu/\~quanrud2}. Supported in part by NSF
    grant CCF-1526799.}}

\maketitle

\newcommand{\defterm}[1]{{\boldmath\normalfont \bfseries #1}}

\newcommand{\distributions}{\Delta}
\newcommand{\diagonal}{\operatorname{diag}\optpar}

\begin{abstract}
  In the regime of bounded transportation costs, additive
  approximations for the optimal transport problem are reduced (rather
  simply) to relative approximations for positive linear programs,
  resulting in faster additive approximation algorithms for optimal
  transport.
\end{abstract}

\section{Introduction}

For $\ell \in \naturalnumbers$, let
$\distributions^\ell = \setof{p \in \nnreals^\ell \where \norm{p}_1 =
  1}$ denote the convex set of probability distributions over
$[\ell]$.  In the \defterm{(discrete) optimal transport problem}, one
is given two distributions $p \in \distributions^\ell$ and
$q \in \distributions^k$ and a nonnegative matrix of
\defterm{transportation costs} $C \in \nnreals^{k \times \ell}$. The
goal is to
\begin{align*}
  \text{minimize }  \sum_{i=1}^k\sum_{j=1}^{\ell} C_{ij} X_{ij} p_j %
  \text{ over } X \in \nnreals^{k \times \ell}                                                        %
  \text{ s.t.\ } X p = q \text{ and } X^t \ones = \ones.    %
  \labelthisequation{T}{transport}
\end{align*}
We let \refequation{transport} denote both the above optimization
problem and its optimal value. \refequation{transport} can be
interpreted as the minimum cost of ``transporting'' a discrete
distribution $p$ to a target distribution $q$, where the cost of
moving probability mass from one coordinate to another is given by
$C$. \refequation{transport} is sometimes called the \emph{earth mover
  distance} between $p$ and $q$, where one imagines $p$ and $q$ as
each dividing the same amount of sand into various piles, and the goal
is to rearrange the piles of sand of $p$ into the piles of sand of $q$
with minimum total effort.

Optimal transport (in much greater generality) is fundamental to
applied mathematics \citep{villani-03,villani-09}. Computing (or
approximating) the optimal transport matrix and its cost has many
applications: we refer to recent work by \citet{cuturi}, \citet{awr}
and \citet{dgk} for further (and up-to-date) references.

Optimal transport is a linear program (abbr.\ LP) and can be solved
exactly by linear program solvers. \refequation{transport} can also be
cast as a minimum cost flow problem, thereby solved combinatorially.
The fastest exact algorithm runs in $\apxO{k \ell \sqrt{k+\ell}}$ time
via minimum cost flow \citep{ls}. Here and throughout $\apxO{\cdot}$
hides polylogarithmic terms in $k,\ell$.

There is recent interest, sparked by \citet{cuturi}, in obtaining
\emph{additive approximations} to \refequation{transport} with running
times that are \emph{nearly linear} in the size of the cost matrix
$C$. For $\delta > 0$, a matrix $X$ is a \defterm{$\delta$-additive
  approximation} if it is a feasible solution to
\refequation{transport} with cost at most a $\delta$ additive factor
more than the optimal transport cost \refequation{transport}. A
``nearly linear'' running time is one whose dependence on $k$ and
$\ell$ is of the form $\bigO{k\ell \polylog{k,\ell}}$; i.e., linear in
the input size up to polylogarithmic factors. \citeauthor{cuturi}
highlights applications in machine learning with large,
high-dimensional datasets, for which a faster approximation algorithm
may be preferable to a slower exact algorithm.

The first nearly linear time additive approximation was obtained
recently by \citet{awr}.  Their result combines a reduction to matrix
scaling observed by \citeauthor{cuturi} and an improved analysis for a
classical matrix scaling algorithm due to \citet{sk} as applied to
this setting (see also \citep{ck}). The bound has a cubic dependency
on $\infnorm{C} / \delta$, where $\infnorm{C} = \max_{i,j} C_{ij}$ is
the maximum value of any coordinate in $C$ and is considered a lower
order term. One factor of $1/\delta$ can be removed by recent advances
in matrix scaling \citep{cmtv} (per \citet{awr}).  A tighter analysis
by \citet{dgk} of the Sinkhorn-Knopp approach decreases the dependency
on $\infnorm{C} / \delta$ to the following.
\begin{theorem}[\citealp{dgk}]
  A $\delta$-additive approximation to \refequation{transport} can be
  computed in $\apxO{k\ell \prac{\infnorm{C}}{\delta}^2}$ time.
\end{theorem}

\subsection{Results}

The optimal transport cost can be approximated more efficiently as
follows. Some of the results are parametrized by the quantity
$\ripover{C}{p}{q}$ instead of $\infnorm{C}$. The quantity
$\ripover{C}{p}{q}$ is the average cost coefficient as sampled from
the product distribution $p \times q$. Needless to say, the average
cost $\ripover{C}{p}{q}$ is at most the maximum cost $\infnorm{C}$,
and the relative difference may be arbitrarily large.
\begin{theorem}
  \labeltheorem{apx-transport-matrix}
  One can compute a $\delta$-additive approximate transportation
  matrix $X$ from $p$ to $q$ sequentially in either
  \begin{enumerate}\raggedright
  \item
    \begin{math}
      \apxO{k\ell \prac{\ripover{C}{q}{p}}{\delta}^2}
    \end{math}
    deterministic time or
  \item
    \begin{math}
      \apxO{k \ell \frac{\infnorm{C}}{\delta}}
    \end{math}
    randomized time;
  \end{enumerate}
  or deterministically in parallel with
  \begin{enumerate}[resume]\raggedright
  \item $\apxO{\prac{\ripover{C}{q}{p}}{\delta}^3}$ depth and
    $\apxO{k\ell \prac{\ripover{C}{p}{q}}{\delta}^2}$ total work, or
  \item
    $\apxO{\prac{\infnorm{C}}{\delta}^2}$ depth and
    $\apxO{k\ell \prac{\infnorm{C}}{\delta}^2}$ total work.
  \end{enumerate}
\end{theorem}

The bounds are obtained rather simply by reducing to a variety of
relative approximation algorithms for certain types of LPs. The
reductions can be summarized briefly as follows.

A simple but important observation is that the transportation cost
from $p$ to $q$ is bounded above by $\ripover{C}{q}{p}$
(\reflemma{oblivious-transport} below). Consequently,
$\epspm$-multiplicative approximations to the value of
\refequation{transport}, for $\eps = \delta / \ripover{C}{q}{p}$, are
$\delta$-additive approximations as well.

The approximate LP solvers produce matrices $X$ that certify the
approximate value, but do not meet the constraints of
\refequation{transport} exactly. In particular, the approximations $X$
transport $\epsless$-fraction of the mass, leaving $\eps$-fraction
behind. The remaining $\eps$-fraction of probability mass is then
transported by a simple oblivious transportation scheme.

Here the algorithms diverge into two types, depending on how to model
\refequation{transport} as an LP. The first approach takes
\refequation{transport} as is, which is a ``positive LP''. Positive
LPs are a subclass of LPs where all coefficients and variables are
nonnegative. Positive LPs can be approximated faster than general LPs
can be solved.  Applied to \refequation{transport}, the approximation
algorithms for positive LPs produce what we call ``$\epsless$-uniform
transportation matrices'', which not only transport all but an
$\eps$-fraction of the total mass, but transport all but an
$\eps$-fraction of each coordinate of $p$, and fill all but an
$\eps$-fraction of each coordinate of $q$. It is shown that
$\epsless$-uniformly approximate transportation matrix can be altered
into exact transportation matrices with an additional cost of about
$\eps \ripover{C}{q}{p}$.

The second approach reformulates \refequation{transport} as a
``packing LP''. Packing LPs are a subclass of positive LPs
characterized by having only packing constraints. The advantage of
packing LPs is that they can be approximated slightly faster than the
broader class of positive LPs. However, the approximate transportation
matrices $X$ produced by the packing LP are not uniformly approximate
in the sense discussed above. Consequently, there is a larger cost of
about $\eps \infnorm{C}$ to extend $X$ to an exact transportation
matrix.

\subsection{Additional background}

There is a burgeoning literature on parametrized regimes of optimal
transport. The many parametrized settings are beyond the scope of this
note, and we refer again to \citep{awr,dgk} for further discussion.

An important special case of the optimal transport problem
\refequation{transport} is where $C$ is a metric or, more generally,
the shortest path metric of an undirected weighted graph. This setting
is equivalent to \emph{uncapacitated minimum cost flow}. Let $m$
denote the number of edges and $n$ the number of vertices of the
underlying graph. Recently, \citet{sherman} proved that a
$\epsmore$-multiplicative approximation to \refequation{transport} can
be obtained in $\apxO{m^{1+o(1)} / {\eps^2}}$ time. This translates to
a $\delta$-additive approximation in
\begin{math}
  \apxO{m^{1+o(1)} \parof{\ripover{C}{q}{p} / {\delta}}^2}
\end{math}
time. Remarkably, if the graph is sparse, then $\apxO{m^{1 + o(1)}}$
is much smaller than the explicit size of the shortest path metric,
$n^2$ -- let alone the time required to compute all pairs of shortest
paths.

There are many applications where the cost matrix $C$ is induced by
some combinatorial or geometric context and may be specified more
sparsely than as $\bigO{k \ell}$ explicit coordinates. It is well
known that some of the LP solvers used below as a black box, as well
as other similar algorithms, can often be extended to handle such
implicit matrices so long as one can provide certain simple oracles
(e.g., \citealp{ky,young,mrwz,cq-17,cq-18}).

The running time (2) of \reftheorem{apx-transport-matrix} was obtained
independently by \citet{bjks-18}, by a similar reduction to packing
LPs. \citet{bjks-18} also get the running time (2) via matrix scaling,
more in the spirit of the preceding works \citep{cuturi,awr,dgk}.

\subsection{Organization}

The rest of this note is organized as follows.
\refsection{oblivious-transport} outlines a simple and crude
approximation algorithm for \refequation{transport}, which is used to
repair approximate transportation matrices, and to upper bound
\refequation{transport}. \refsection{transport-pc} applies positive LP
solvers to approximate \refequation{transport}, and leads to the
running times in \reftheorem{apx-transport-matrix} that depend on
$\ripover{C}{q}{p}$ and not
$\infnorm{C}$. \refsection{transport-packing} applies packing LP
solvers to a reformulation of \refequation{transport}. This approach
leads to the remaining running times in
\reftheorem{apx-transport-matrix} that all depend on $\infnorm{C}$.

\section{Oblivious transport}

\labelsection{oblivious-transport}

The high-level idea is to use approximate LP solvers to transport most
of $p$ to $q$, and then transport the remaining probability mass with
a cruder approximation algorithm. The second step always uses the
following oblivious transportation scheme. The upper bound obtained
below also provides a frame of reference for comparing additive and
relative approximation factors, and is useful for bounding a binary
search for the optimal value.
\begin{lemma}
  \labellemma{oblivious-transport} For a distribution
  $q \in \distributions^k$, consider the matrix
  $X \in \nnreals^{k \times \ell}$ with each column set to $q$; i.e.,
  $X_{ij} = q_i$ for all $i,j$. For any $p \in \distributions^\ell$,
  $X$ is a transportation matrix from $p$ to $q$, with total cost
  \begin{math}
    \ripover{C}{q}{p}.
  \end{math}
\end{lemma}
\begin{proof}
  Fix $p \in \distributions^\ell$.  For any $i \in [m]$,
  \begin{align*}
    \ripover{X}{e_i}{p}         %
    =                           %
    \sum_{j=1}^\ell X_{ij} p_j        %
    =                             %
    q_i \sum_{j=1}^\ell p_j          %
    \tago{=}                             %
    q_i
  \end{align*}
  since \tagr $p$ is a distribution.  For any $j \in [n]$, we have
  \begin{align*}
    \ripover{X}{\ones}{e_j}     %
    =                           %
    \sum_{i=1}^k X_{ij}         %
    =                           %
    \sum_{i=1}^k q_i            %
    \tago{=}
    1.
  \end{align*}
  since \tagr $q$ is a distribution.  Thus $X$ is a transportation
  matrix from $q$ to $p$. The transportation cost of $X$ is
  \begin{align*}
    \sum_{i=1}^k \sum_{j=1}^\ell C_{ij} X_{ij} p_j %
    =                                           %
    \sum_{i=1}^k \sum_{j=1}^\ell C_{ij} q_i p_j
    =                           %
    \ripover{C}{q}{p},
  \end{align*}
  as desired.
\end{proof}

\section{Reduction to mixed packing and covering}

\labelsection{transport-pc}

Our first family of approximation algorithms, which obtain the bounds
in \reftheorem{apx-transport-matrix} that are relative to
$\ripover{C}{q}{p}$, observe that optimal transport lies in the
following class of LPs. A \defterm{mixed packing and covering program}
is a problem of any of the forms
\begin{align*}
  \setof{\text{find } x, \text{ max } \rip{v}{p}, \text{ or min }
  \rip{v}{p}}
  \text{ over }
  x \in \nnreals^n
  \text{ s.t.\ } A x \leq b %
  \text{ and } C x \geq d, \labelthisequation{PC}{pc}
\end{align*}
where $A \in \nnreals^{m_1 \times n}$, $b \in \nnreals^{m_1}$,
$C \in \nnreals^{m_2 \times n}$, and $d \in \nnreals^{m_2}$, and
$v \in \nnreals^n$ all have nonnegative coefficients. We let $N$
denote the total number of nonzeroes in the input.  For $\eps > 0$, an
\defterm{$\eps$-relative approximation} to \refequation{pc} is either
(a) a certificate that \refequation{pc} is either infeasible, or (b) a
nonnegative vector $x \in \nnreals^n$ such that $A x \leq \epsmore b$
and $C x \geq \epsless d$ and, when there is a linear objective and
the linear program is feasible, within a $\epspm$-multiplicative
factor of the optimal value. Relative approximations to positive LPs
can be obtained with nearly-linear dependence on $N$, and polynomial
dependency on $\frac{1}{\eps}$, as follows.
\begin{lemma}[\citealp{young}]
  \labellemma{deterministic-pc}
  Given an instance of \refequation{pc} and $\eps > 0$, one can
  compute a $\eps$-relative approximation to \refequation{pc} in
  \begin{math}
    \apxO{N / \eps^2}
  \end{math}
  deterministic time.
\end{lemma}
\begin{lemma}[\citealp{mrwz}]
  \labellemma{parallel-pc} Given an instance of a mixed packing and
  covering problem \refequation{pc} and $\eps > 0$, one can compute a
  $\eps$-relative approximation to \refequation{pc} deterministically
  in parallel in
  \begin{math}
    \apxO{1/\eps^3}
  \end{math}
  depth and
  \begin{math}
    \apxO{N/\eps^2}
  \end{math}
  total work.
\end{lemma}

\refequation{transport} is a minimization instance of mixed packing
and covering that is always feasible. The role of nonnegative
variables is played by the coordinates of the transportation matrix
$X \in \nnreals^{n \times n}$, with costs $p_j C_{ij}$ for each
$X_{ij}$. The two equations $X p = q$ and $X^t \ones = \ones$ each
give rise to two sets of packing constraints, $X p \leq q$ and
$X^t \ones \leq \ones$, and two sets of covering constraints,
$X p \geq q$ and $X^t \ones \geq \ones$. We have $N = \bigO{k \ell}$
nonzeroes, $m = 2 (k + \ell)$ packing and convering constraints, and
$n = k \ell$ variables.

An $\eps$-relative approximation to \refequation{pc} is not
necessarily a feasible solution to \refequation{transport}. To help
characterize the difference, we define the following. For fixed
$\eps > 0$ and two distributions $p \in \distributions^\ell$ and
$q \in \distributions^k$, a \defterm{$\epsless$-uniform transportation
  matrix} from $p$ to $q$ is a nonnegative matrix
$X \in \nnreals^{k \times \ell}$ with $\epsless q \leq X p \leq q$ and
$\epsless \ones \leq X^t \ones \leq \ones$.

\begin{lemma}
  \labellemma{apx-uniform-transport}
  Given an instance of the optimal transport problem
  \refequation{transport} and $\eps > 0$, a $\epsless$-uniform
  transport matrix with cost at most \refequation{transport} can be
  computed
  \begin{enumerate}
  \item sequentially in
    \begin{math}
      \apxO{\frac{kl}{\eps^2}}
    \end{math}
    time, and
  \item in parallel in
    \begin{math}
      \apxO{\frac{1}{\eps^3}}
    \end{math}
    depth and
    \begin{math}
      \apxO{\frac{kl}{\eps^2}}
    \end{math}
    total work.
  \end{enumerate}
\end{lemma}
\begin{proof}
  By either \reflemma{deterministic-pc} or \reflemma{parallel-pc}, one
  can compute an $\eps$-approximation $X$ to \refequation{transport}
  with the claimed efficiency. Then $\epsless X$ is a
  $(1-\eps)^2$-uniform approximate transportation matrix.
\end{proof}

\begin{lemma}
  \labellemma{fix-apx-uniform-transport} Given a $\epsless$-uniform
  approximate transportation matrix $X$, one can compute a
  transportation matrix $U$ with cost at most an additive factor of
  \begin{math}
    4 \eps \ripover{C}{q}{p}
  \end{math}
  more than the cost of $X$, in linear time and work and with constant
  depth.
\end{lemma}

\begin{proof}
  We first scale down $X$ slightly to a $(1-2\eps)$-uniform
  transportation matrix, and then augment the shrunken transportation
  matrix with the oblivious transportation scheme from
  \reflemma{oblivious-transport}.  Clearly this can be implemented in
  linear time and work and in constant depth.

  Let $Y = \parof{1 - \frac{\eps}{1 - \eps}} X$.  Then $Y$ is a
  $(1 - 2 \eps)$-uniform approximate transportation matrix from $p$ to
  $q$. Let $p' = \parof{I - \diagonal{S^t \ones}} p$ and
  $q' = q - S p$. Since $Y$ is $(1-2\eps)$-uniform, we have
  $p' \leq 2\eps p$ and $q' \leq 2 \eps q$.  $p'$ represents the
  probability mass not yet transported by $Y$, and $q'$ represents the
  probability mass not yet filled by $Y$, and we have
  \begin{align*}
    \rip{\ones}{p'} = \rip{\ones}{q'}.
  \end{align*}
  Let $\alpha$ denote this common value. Then
  \begin{align*}
    \alpha = \rip{\ones}{q'}    %
    =
    1 - \rip{\ones}{Y x}        %
    \geq                        %
    1 - \rip{\ones}{X p} + \frac{\eps}{1-\eps} \rip{\ones}{X p} %
    \tago{\geq}                                           %
    \eps
  \end{align*}
  because \tagr $X$ being $\epsless$-uniform implies
  $1 - \eps \leq \rip{X p}{\ones} \leq 1$. Let $Z$ be the matrix where
  each column is $q' / \alpha$; by \reflemma{oblivious-transport}, $Z$
  is a transportation matrix from $p'/\alpha$ to $q' / \alpha$. Let
  $Z' = Z (I - \diagonal{Y^t \ones})$. Then
  \begin{align*}
    (Y + Z') p = Y p + Z p' = Y x + q' = q,
  \end{align*}
  and
  \begin{align*}
    (Y + Z')^t \ones            %
    &= Y^t \ones + (I - \diagonal{Y^t \ones}) Z^t \ones %
        =     %
        Y^t \ones + (I - \diagonal{Y^t \ones}) \ones
    \\
    &=                           %
      Y^t \ones + \ones - Y^t \ones = \ones,
  \end{align*}
  so $Y + Z'$ is a transportation matrix from $p$ to $q$. The cost of
  $Y$ is less than the cost of $X$, and the cost of $Z'$ is at most
  \begin{align*}
    \sum_{i,j} C_{ij} Z'_{ij} p_j %
    &\tago{=}                           %
      \sum_{i,j} C_{ij} Z_{ij} p'_j %
      \tago{=}                             %
      \frac{1}{\alpha} \sum_{i,j} C_{ij} q'_i p'_j %
    \\
    &=                                            %
      \frac{1}{\alpha} \ripover{C}{q'}{p'}         %
      \tago{\leq}                                         %
      \frac{4 \eps^2}{\alpha} \ripover{C}{q}{p} %
    \\
    &\leq                                           %
      4 \eps \ripover{C}{q}{p}
  \end{align*}
  by \tagr definition of $Z'$, \tagr definition of $Z$, \tagr
  $p' \leq 2\eps p$ and $q' \leq 2 \eps q$, and \tagr
  \begin{math}
    \alpha \geq \eps.
  \end{math}
\end{proof}

\begin{theorem}
  One can deterministically compute a $\delta$-additive approximation to
  \refequation{transport}
  \begin{enumerate}
  \item sequentially in
    \begin{math}
      \apxO{k l \prac{\ripover{C}{x}{y}}{\delta}^2}
    \end{math}
    time, and
  \item in parallel in $\apxO{\prac{\ripover{C}{x}{y}}{\delta}^3}$
    depth and $\apxO{k l \prac{\ripover{C}{x}{y}}{\delta}^2}$
    total work.
  \end{enumerate}
\end{theorem}
\begin{proof}
  Given $\delta > 0$, let $\eps = \frac{\delta}{4
    \ripover{C}{q}{p}}$. We apply \reflemma{apx-uniform-transport} to
  generate a $\epsless$-uniform transportation matrix $X$ of cost at
  most \refequation{transport} within the desired time/depth
  bounds. We then apply \reflemma{fix-apx-uniform-transport} to $X$ to
  construct a transportation matrix $U$ with cost at most
  \begin{math}
    \refequation{transport} + 4 \eps \ripover{C}{q}{p} %
    =                                                    %
    \refequation{transport} + \delta,
  \end{math}
  as desired.
\end{proof}

\section{Reduction to packing}
\labelsection{transport-packing}

A \defterm{(pure) packing LP} is a linear program of the form
\begin{align}
  \text{maximize } \rip{c}{p}   %
  \text{ over } x \in \nnreals^n %
  \text{ over } A x \leq b, %
  \labelthisequation{P}{packing-lp}
\end{align}
where $A \in \nnreals^{m \times n}$, $b \in \nnreals^m$, and
$c \in \nnreals^{n}$.  For a fixed instance of
\refequation{packing-lp}, we let $N$ denote the number of nonzeroes in
$A$. For $\eps > 0$, a \defterm{$\epsless$-relative approximation} to
\refequation{packing-lp} is a point $x \in \nnreals^n$ such that
$A x \leq b$ and $\rip{c}{x}$ is at least $\epsless$ times the optimal
value of \refequation{packing-lp}.  $\epsless$-relative approximations
packing LPs can be obtained slightly faster than $\eps$-relative
approximations to more general positive linear programs, as follows.
\begin{lemma}[\citealp{ao-15-stoc}]
  \labellemma{ao} Given an instance of the pure packing problem
  \refequation{packing-lp}, and $\eps > 0$, a
  $\epsless$-multiplicative approximation to \refequation{packing-lp}
  can be computed in $\apxO{N / \eps}$ randomized time.
\end{lemma}
\begin{lemma}[\citealp{mrwz}]
  \labellemma{parallel-packing} Given an instance of the pure packing
  problem \refequation{packing-lp}, and $\eps > 0$, a
  $\epsless$-multiplicative approximation to \refequation{packing-lp}
  can be computed deterministically in parallel in $\apxO{1 / \eps^2}$
  depth and $\apxO{N / \eps^2}$ total work.
\end{lemma}

Consider the following LP reformulation of \refequation{transport},
that is parametrized by a value $\lambda$ that specifies a desired
transportation cost.
\begin{align}
  \begin{aligned}
    \text{max } %
    & \ripover{X}{\ones}{p} \text{ over } X \in \nnreals^{k \times \ell} \\
    \text{s.t.\ } %
    & \sum_{i=1}^k X_{ij} \leq 1 \text{ for all } j, \\ %
    & \sum_{j=1}^\ell X_{ij} p_j \leq q_i \text{ for all } i \in [m], \\
    & \sum_{i=1}^k \sum_{j=1}^\ell C_{ij} X_{ij} p_j \leq \lambda.
  \end{aligned}
      \labelthisequation{TP$(\lambda)$}{transport-packing}
\end{align}
The advantage of \refequation{transport-packing} compared to
\refequation{transport} is that \refequation{transport-packing} is a
packing LP, which as observed above can be solved slightly faster than
a mixed packing and covering LP. The packing problem
\refequation{transport-packing} has $m = \bigO{k + \ell}$ packing
constraints, $n = k\ell$ variables, and $N = \bigO{k \ell}$ nonzeroes.

$\epsless$-approximations to \refequation{transport-packing} are not
feasible solutions to \refequation{transport} even for
$\lambda = \refequation{transport}$. To help characterize the
difference, we define the following.  For fixed $\eps > 0$ and two
distributions $p \in \distributions^\ell$ and
$q \in \distributions^k$, a \defterm{$\epsless$-transportation matrix
  from $p$ to $q$} is a nonnegative matrix
$X \in \nnreals^{k \times \ell}$ with $X p \leq q$,
$X^t \ones \leq \ones$, and $\rip{\ones}{X p} \geq 1 - \eps$.

\begin{lemma}
  \labellemma{apx-packing-transport} Consider an instance of the
  optimal transport problem \refequation{transport}, and let
  $\eps, \delta > 0$ be fixed parameters with
  $\eps \leq \frac{\delta}{\ripover{C}{p}{q}}$. One can compute an
  $\epsless$-transportation matrix from $p$ to $q$ with cost at most
  $\refequation{transport} + \delta$
  \begin{enumerate}
  \item sequentially in
    \begin{math}
      \apxO{\frac{kl}{\eps}}
    \end{math}
    randomized time, and
  \item in parallel in
    \begin{math}
      \apxO{\frac{1}{\eps^2}}
    \end{math}
    depth and
    \begin{math}
      \apxO{\frac{kl}{\eps^2}}
    \end{math}
    total work.
  \end{enumerate}
\end{lemma}
\begin{proof}
  For fixed $\lambda$, either $\lambda \leq \refequation{transport}$,
  or either $\epsless$-approximation algorithm from \reflemma{ao} or
  \reflemma{parallel-packing} returns $\epsless$-transportation
  matrices $X$ from $p$ to $q$ with cost at most $\lambda$.  We wrap
  the $\epsless$-relative approximation algorithms in a binary search
  for the smallest value, up to an additive factor of $\delta$, that
  produces a $\epsless$-transportation matrix from $p$ to $q$. Since
  $\lambda = \refequation{transport}$ is sufficient, such a search
  returns a value of
  \begin{math}
    \lambda \leq \refequation{transport} + \delta.
  \end{math}
  By \reflemma{oblivious-transport}, the search can be bounded to the
  range $[0, \ripover{C}{q}{p}]$. Thus the binary search needs at most
  $\bigO{\log{\frac{\ripover{C}{q}{p}}{\delta}}}$ iterations to
  identify such a value $\lambda$, for which we obtain the desired
  $\epsless$-transportation matrix. We can assume that
  $\frac{\ripover{C}{p}{q}}{\delta}$ is at most $\poly{k,\ell}$, since
  otherwise \refequation{transport} can be solved exactly in
  \begin{math}
    \poly{k,\ell} \leq \frac{\ripover{C}{p}{q}}{\delta} \leq
    \frac{1}{\eps}
  \end{math}
  time.
\end{proof}

\begin{lemma}
  \labellemma{fix-apx-transport} Let $X$ be a
  $\epsless$-transportation matrix from $p$ to $q$. In $\bigO{k \ell}$
  time, one can extend $X$ to a transportation matrix $U$ with an
  additional cost of $\eps \infnorm{C}$.
\end{lemma}
\begin{proof}
  We use the oblivious transportation scheme of
  \reflemma{oblivious-transport} to transport the remaining
  $\eps$-fraction of mass. It is straightforward to verify the
  additional transportation costs at most $\eps \infnorm{C}$, as
  follows.

  Let $p' = (I - \diagonal{X^t \ones}) p$ and $q' = q - X p$.  $p'$
  represents the probability mass not yet transported by $X$, and $q'$
  represents the probability mass not yet transported by $Y$. Let
  \begin{math}
    \alpha = \rip{\ones}{p'} = \rip{\ones}{q'}
  \end{math}
  denote the residual probability mass.  Since $X$ is a
  $\epsless$-transportation matrix, we have
  \begin{align*}
    \alpha = \rip{\ones}{q'} = 1 - \rip{\ones}{X p} \leq \eps.
    \numberthis
  \end{align*}
  Let $Y$ be the matrix\footnote{In fact, any transportation matrix
    from $p'$ to $q'$ will do.} where each column is set to
  $q' / \alpha$; by \reflemma{oblivious-transport}, $Y$ is a
  transportation matrix from $p' / \alpha$ to $q'/\alpha$. Let
  $Y' = Y (I - \diagonal{X^t \ones}) p$. By the same calculations as
  in the proof of \reflemma{fix-apx-uniform-transport}, $X + Y'$ is a
  transportation matrix from $q$ to $p$. The cost of $Y'$ is
  \begin{align*}
    \sum_{ij} C_{ij} Y'_{ij} p_j %
    \tago{=}
    \frac{1}{\alpha}    %
    \ripover{C'}{q'}{p'}        %
    \tago{\leq}                        %
    \frac{\norm{q'}_1 \norm{p'}_1}{\alpha} \infnorm{C'} %
    =                                                %
    \alpha \infnorm{C'} \tago{\leq} \eps \infnorm{C'}
  \end{align*}
  by \tagr the proof of \reflemma{fix-apx-uniform-transport}, \tagr
  Cauchy-Schwartz, and \tagr the above inequality \reflastequation.
\end{proof}

\begin{theorem}
  One can compute a $\delta$-additive approximation to
  \refequation{transport}
  \begin{enumerate}\raggedright
  \item sequentially in
    $\apxO{k\ell \frac{\infnorm{C}}{\delta}}$ randomized time, and
  \item in parallel with
    \begin{math}
      \apxO{ \prac{\infnorm{C}}{\delta}^2}
    \end{math}
    depth and
    \begin{math}
      \apxO{ k\ell \prac{\infnorm{C}}{\delta}^2}
    \end{math}
    total work.
  \end{enumerate}
\end{theorem}

\begin{proof}
  Let $\delta > 0$ be fixed. Let
  $\eps = \frac{\delta}{2 \infnorm{C}}$.  By
  \reflemma{apx-packing-transport}, we can compute a
  $\epsless$-transportation matrix $X$ with cost at most
  \begin{math}
    \refequation{transport} + \frac{\delta}{2}.
  \end{math}
  By \reflemma{fix-apx-transport}, we can extend $X$ to a
  transportation matrix $U$ with additional cost of at most
  \begin{math}
    \eps \infnorm{C} = \frac{\delta}{2},
  \end{math}
  for a total cost at most $\refequation{transport} + \delta$.
\end{proof}

\paragraph{Acknowledgements.} We thank Jason Altschuler and Chandra
Chekuri for insightful discussions and helpful feedback.

\begingroup
\let\stdsection\section
\def\section*#1{\stdsection{#1}}
\bibliographystyle{plainnat} %
\bibliography{transport-lp} %

\begin{thebibliography}{17}
\providecommand{\natexlab}[1]{#1}
\providecommand{\url}[1]{\texttt{#1}}
\expandafter\ifx\csname urlstyle\endcsname\relax
  \providecommand{\doi}[1]{doi: #1}\else
  \providecommand{\doi}{doi: \begingroup \urlstyle{rm}\Url}\fi

\bibitem[{Allen Zhu} and Orecchia(2015)]{ao-15-stoc}
Zeyuan {Allen Zhu} and Lorenzo Orecchia.
\newblock Nearly-linear time positive {LP} solver with faster convergence rate.
\newblock In \emph{Proceedings of the Forty-Seventh Annual {ACM} on Symposium
  on Theory of Computing, {STOC} 2015, Portland, OR, USA, June 14-17, 2015},
  pages 229--236, 2015.

\bibitem[Altschuler et~al.(2017)Altschuler, Weed, and Rigollet]{awr}
Jason Altschuler, Jonathan Weed, and Philippe Rigollet.
\newblock Near-linear time approximation algorithms for optimal transport via
  {Sinkhorn} iteration.
\newblock In \emph{Advances in Neural Information Processing Systems 30: Annual
  Conference on Neural Information Processing Systems 2017, 4-9 December 2017,
  Long Beach, CA, {USA}}, pages 1961--1971, 2017.

\bibitem[Blanchet et~al.(2018)Blanchet, Jambulapati, Kent, and
  Sidford]{bjks-18}
Jose Blanchet, Arun Jambulapati, Carson Kent, and Aaron Sidford.
\newblock Towards optimal running times for optimal transport.
\newblock \emph{CoRR}, abs/1810.07717, 2018.
\newblock URL \url{http://arxiv.org/abs/1810.07717}.

\bibitem[Chakrabarty and Khanna(2018)]{ck}
Deeparnab Chakrabarty and Sanjeev Khanna.
\newblock Better and simpler error analysis of the {Sinkhorn}-{Knopp} algorithm
  for matrix scaling.
\newblock In \emph{1st Symposium on Simplicity in Algorithms, {SOSA} 2018,
  January 7-10, 2018, New Orleans, LA, {USA}}, pages 4:1--4:11, 2018.

\bibitem[Chekuri and Quanrud(2017)]{cq-17}
Chandra Chekuri and Kent Quanrud.
\newblock Near-linear time approximation schemes for some implicit fractional
  packing problems.
\newblock In \emph{Proceedings of the Twenty-Eighth Annual {ACM-SIAM} Symposium
  on Discrete Algorithms, {SODA} 2017, Barcelona, Spain, Hotel Porta Fira,
  January 16-19}, pages 801--820, 2017.

\bibitem[Chekuri and Quanrud(2018)]{cq-18}
Chandra Chekuri and Kent Quanrud.
\newblock Randomized {MWU} for positive {LP}s.
\newblock In \emph{Proceedings of the Twenty-Ninth Annual {ACM-SIAM} Symposium
  on Discrete Algorithms, {SODA} 2018, New Orleans, LA, USA, January 7-10,
  2018}, pages 358--377, 2018.

\bibitem[Cohen et~al.(2017)Cohen, Madry, Tsipras, and Vladu]{cmtv}
Michael~B. Cohen, Aleksander Madry, Dimitris Tsipras, and Adrian Vladu.
\newblock Matrix scaling and balancing via box constrained {Newton}'s method
  and interior point methods.
\newblock In \emph{58th {IEEE} Annual Symposium on Foundations of Computer
  Science, {FOCS} 2017, Berkeley, CA, USA, October 15-17, 2017}, pages
  902--913, 2017.

\bibitem[Cuturi(2013)]{cuturi}
Marco Cuturi.
\newblock {Sinkhorn} distances: Lightspeed computation of optimal transport.
\newblock In \emph{Advances in Neural Information Processing Systems 26: 27th
  Annual Conference on Neural Information Processing Systems 2013. Proceedings
  of a meeting held December 5-8, 2013, Lake Tahoe, Nevada, United States.},
  pages 2292--2300, 2013.

\bibitem[Dvurechensky et~al.(2018)Dvurechensky, Gasnikov, and Kroshnin]{dgk}
Pavel Dvurechensky, Alexander Gasnikov, and Alexey Kroshnin.
\newblock Computational optimal transport: Complexity by accelerated gradient
  descent is better than by {Sinkhorn}'s algorithm.
\newblock In \emph{Proceedings of the 35th International Conference on Machine
  Learning, {ICML} 2018, Stockholmsm{\"{a}}ssan, Stockholm, Sweden, July 10-15,
  2018}, pages 1366--1375, 2018.

\bibitem[Koufogiannakis and Young(2014)]{ky}
Christos Koufogiannakis and Neal~E. Young.
\newblock A nearly linear-time {PTAS} for explicit fractional packing and
  covering linear programs.
\newblock \emph{Algorithmica}, 70\penalty0 (4):\penalty0 648--674, 2014.
\newblock Preliminary version in FOCS 2007.

\bibitem[Lee and Sidford(2014)]{ls}
Yin~Tat Lee and Aaron Sidford.
\newblock Path finding methods for linear programming: Solving linear programs
  in $\tilde{O}(\sqrt{rank})$ iterations and faster algorithms for maximum
  flow.
\newblock In \emph{55th {IEEE} Annual Symposium on Foundations of Computer
  Science, {FOCS} 2014, Philadelphia, PA, USA, October 18-21, 2014}, pages
  424--433, 2014.

\bibitem[Mahoney et~al.(2016)Mahoney, Rao, Wang, and Zhang]{mrwz}
Michael~W. Mahoney, Satish Rao, Di~Wang, and Peng Zhang.
\newblock Approximating the solution to mixed packing and covering {LP}s in
  parallel $\tilde{O}(\epsilon^{-3})$ time.
\newblock In \emph{43rd International Colloquium on Automata, Languages, and
  Programming, {ICALP} 2016, July 11-15, 2016, Rome, Italy}, pages 52:1--52:14,
  2016.

\bibitem[Sherman(2017)]{sherman}
Jonah Sherman.
\newblock Generalized preconditioning and undirected minimum-cost flow.
\newblock In \emph{Proceedings of the Twenty-Eighth Annual {ACM-SIAM} Symposium
  on Discrete Algorithms, {SODA} 2017, Barcelona, Spain, Hotel Porta Fira,
  January 16-19}, pages 772--780, 2017.

\bibitem[Sinkhorn and Knopp(1967)]{sk}
Richard Sinkhorn and Paul Knopp.
\newblock Concerning nonnegative matrices and doubly stochastic matrices.
\newblock \emph{Pacific Journal of Mathematics}, 21\penalty0 (2):\penalty0
  343--348, 1967.

\bibitem[Villani(2003)]{villani-03}
Cédric Villani.
\newblock \emph{Topics in Optimal Transportation}, volume~58 of \emph{Graduate
  Studies in Mathematics}.
\newblock American Mathematical Society, 2003.

\bibitem[Villani(2009)]{villani-09}
Cédric Villani.
\newblock \emph{Optimal transport: old and new}, volume 338 of
  \emph{Grundlehren der mathematischen Wissenschaften}.
\newblock Springer, Berlin, Heidelberg, 2009.

\bibitem[Young(2014)]{young}
Neal~E. Young.
\newblock Nearly linear-time approximation schemes for mixed packing/covering
  and facility-location linear programs.
\newblock \emph{CoRR}, abs/1407.3015, 2014.
\newblock URL \url{http://arxiv.org/abs/1407.3015}.

\end{thebibliography}
\endgroup

\end{document}